\documentclass[letterpaper, 10pt, conference]{ieeeconf}
\IEEEoverridecommandlockouts                              

\usepackage{multirow}
\usepackage{lipsum} 
\usepackage{graphics} 
\usepackage{epsfig} 
\usepackage{mathptmx} 
\usepackage{times} 
\usepackage{amsmath} 
\usepackage{amssymb}  
\usepackage[english]{babel}
\usepackage[noadjust]{cite}
\usepackage{epstopdf}
\usepackage{comment}
\usepackage{subcaption}
\newtheorem{theorem}{Theorem}[section]

\newtheorem{assumption}{Assumption}

\title{\LARGE \bf
Adaptive Extremum Seeking Using Recursive Least Squares}

\author{Nursefa Zengin, Baris Fidan
\thanks{{School of Engineering, University of Waterloo, ON, Canada, ~\tt\small{nyarbasi,fidan@uwaterloo.ca}}}
}

\begin{document}

\maketitle
\thispagestyle{empty}
\pagestyle{empty}

\begin{abstract}
Extremum seeking (ES) optimization approach has been very popular due to its non-model based analysis and implementation. This approach has been mostly used with gradient based search algorithms. Since least squares (LS) algorithms are typically observed to be superior, in terms of convergence speed and robustness to measurement noises, over gradient algorithms, it is expected that LS based ES schemes will also provide faster convergence and robustness to sensor noises. In this paper, with this motivation, a recursive least squares (RLS) estimation based ES scheme is designed and analysed for application to scalar parameter and vector parameter static map and dynamic systems. Asymptotic convergence to the extremum is established for all the cases. Simulation studies are provided to validate the performance of proposed scheme.

\end{abstract}

\section{Introduction}

Extremum seeking (ES) is a popular technique for adaptive optimization of the performance of dynamic systems by tuning certain system parameters based on measurements. The main advantage of this technique is that limited or no knowledge of the plant model is required. ES is suitable for optimization of the performance of systems with complex dynamics, unavailable suitable measurements to validate the model, and time-varying disturbances that are difficult to model accurately (\cite{krsticbook}).

The most common ES algorithm used in the literature is the classical band-pass filtering based one, in which the gradient of the output with respect to the input will determine the direction of adjusting the input variables. This method was successfully applied to different application areas including biochemical reactors [\cite{hsin1999,bastin2009}], ABS control in automotive brakes (\cite{drakunov1995,krsticbook,yu2002,dincmen2014,dincmen2014b}), mobile robots (\cite{mayhew2007,zhang2009,lin2017}), mobile sensor networks (\cite{biyik2008,stankovic2009,moore2010}).

Among other types of ES algorithms, perturbation based ES relies on added perturbation signals to estimate the gradient of the output by correlating the perturbations. To overcome the implementation drawbacks of introducing perturbation signals, some methods that are free of perturbation signals have been developed by  \cite{fu_ozguner,lsesc,nesic_nodither}.

Convergence rate of conventional ES algorithms is a limiting factor in many applications. Recursive Least Squares (RLS) based estimation has significant potential in relaxing this limitation and improving robustness to measurement noises.  \cite{escrls2016,lsesc,nesicls} used certain LS based techniques in their ES algorithms to obtain better convergence results. \cite{escrls2016} estimated the gradient of the
output with respect to the input using a LS based adaptive law for a class of nonlinear dynamic systems together with a sinusoidal perturbation signal. \cite{lsesc} used past data of a performance map to estimate the gradient of this performance map by a first order LS fit. The proposed method used no dither signal, but utilized a time window of history data of the performance map. \cite{nesicls} provided general results and a framework for the design of ES schemes applied to systems with parametric uncertainties and used LS algorithm to estimate unknown parameters of the known system.

In absence of the parameter knowledge, a series of control/optimization schemes have been proposed in the literature utilizing certain ES tools such as switching methods (\cite{blackman1962}), signal perturbation for persistence excitation, and band pass filtering (\cite{blackman1962},\cite{krsticextremum},\cite{krsticwang},\cite{nesicls}. \cite{guay2014} and \cite{guay2015} used a discrete time ES scheme to estimate the gradient as a time-varying parameter using LS like update laws. They removed the need for averaging system in order to achieve the convergence of ES. The designs are simulated for static unknown maps, systems with unknown discrete-time dynamics and sampled-data systems. 
 
In this paper, a continuous time RLS parameter estimation based ES scheme is designed and analysed for scalar parameter and vector parameter static map and dynamic systems. Asymptotic convergence to the extremum is established for each case. Numerical simulation examples are provided to validate the performance of proposed scheme comparing the results with gradient parameter estimation based one. A specific simulation example, antilock braking systems (ABS), in \cite{krsticbook} is studied to compare the performance of RLS estimation based ES with classical gradient based ES.

Contents of this paper are as follows. Section II is dedicated to the problem statement. In Section III, existing classical perturbation based ES is reviewed. Proposed RLS estimation based adaptive ES is developed for scalar parameter systems in Section IV, and for vector parameter systems in Section V. Comparative simulation examples are presented in Section VI. Finally, conclusions of the paper are given in Section VII.

\section{Problem Statement}
The ES problem of interest is defined for static map systems and dynamic systems separately in the following subsections. 

\subsection{Static Maps} \label{section1}
Consider a concave static map system
\begin{equation} \label{staticmap_ls}
\begin{aligned}
y=h_s(u)= \bar{h}_s(\theta^*,u), \quad \theta^*=\begin{bmatrix} \theta^*_{1} & \cdots & \theta^*_{N} \end{bmatrix} ^{T},
\end{aligned}
\end{equation}
where $\theta^* \in \mathbb{R}^{N}$ is a fixed unknown parameter vector, $u \in \mathbb{R}^{m}$ is the input and $y \in \mathbb{R}$ is the output of the system. Assume that the control input signal $u$ is generated by a smooth control law 
\begin{equation} \label{systemforlsinput}
u=\alpha(\theta)
\end{equation}
parametrized by a control parameter vector $\theta \in \mathbb{R}^{N}$.
\begin{assumption} \label{assump_1}
The static map  $\bar{h}_s(\theta^*,u)$ is smoothly differentiable. \end{assumption}
\begin{assumption} \label{assump_2}
$h_s(u)= \bar{h}_s(\theta^*,u)$ has a single extremum (maximum) $y^*$ at $u=\alpha(\theta^*).$
\end{assumption}
The control objective is to maximize the steady-state value of $y$ but without requiring the knowledge of $\theta^*$ or the system function $h_s$.
\subsection{Dynamic Systems} \label{section2}
Consider a general multi-input-single-output (MISO) nonlinear system
\begin{equation} \label{systemforls}
\dot{x} = f(x,u) = \bar{f}(\theta^*,x,u), 
\end{equation}
\begin{equation} \label{y_ls}
y = h_d(x)= \bar{h}_d(\theta^*,\theta) = h(\theta),
\end{equation}
\begin{equation} \label{y_ls2}
\theta=\pi(x)
\end{equation}
where $x\in\mathbb{R}^{n}$ is the state, $u\in\mathbb{R}^{m}$ is the input, $y\in\mathbb{R}$ is the output, all measurable, and $f:\mathbb{R}^{n} \times \mathbb{R}^{m}\rightarrow\mathbb{R}^n$ and $h_d=h\circ \pi$ are smooth functions. Assume that the control input signal $u$ is in the form (\ref{systemforlsinput}), the control parameter $\theta \in \mathbb{R}^N$ is dependant on $x$ through a map $\pi(.):\mathbb{R}^n \rightarrow \mathbb{R}^N$.

The closed loop system can be written as follows:
\begin{equation} \label{closed-loop-system}
\dot{x}=f(x, \alpha(\theta)) = f(x,\alpha(\pi(x)).
\end{equation}
The equilibria of (\ref{closed-loop-system}) can be parameterized by $\theta$. The following assumptions about the closed loop system (\ref{systemforls}) are made, similarly to \cite{krsticwang}.
\begin{assumption} \label{assump3}
There exists a smooth function $l:\mathbb{R}^N\rightarrow\mathbb{R}^{m}$ such that
\begin{equation}
f(x, \alpha(x,\theta))=0\quad  \text{if and only if}\quad x=l(\theta),
\end{equation}
for any $(x,\theta) \in \mathbb{R}^m  \times \mathbb{R}^{N}.$ For each $\theta \in \mathbb{R}^N$, the equilibrium $x_e=l(\theta)$ of the system (\ref{closed-loop-system}) is locally exponentially stable with decay and overshoot constants uniformly dependent on $\theta$.
\end{assumption}

\begin{assumption} \label{assump5}
There exists $\theta^{*}\in\mathbb{R}^N$ such that for all admissible $x$ values, $h_d(x)$ has its unique maximum at $x=x^*=l(\theta^*),$
\begin{equation}
y^{\prime} (x^*)={\frac{\partial h}{\partial x}}\Bigr|_{\substack{x=x^*}}= 0, 
\end{equation}
and the $m \times m$ Hessian matrix $y^{\prime\prime} (x^*)= {\frac{\partial^2  h}{\partial x^2}}\Bigr|_{\substack{x=x^*}}$ is negative definite.
\end{assumption}
The control objective is to maximize the steady-state value of $y$ but without requiring the knowledge of $\theta^*$ or the system functions $h_d,f$. This objective could be perfectly performed if $\theta^*$ was known and substituted in (\ref{systemforlsinput}).

The control parameter vector estimation can be done in different ways, leading to different ES schemes, even for the fixed control structure (\ref{systemforlsinput}). The assumption that $h$ has a maximum is without loss of generality, considering a maximum seeking task. Minimum seeking case would be treated identically, replacing $y$ with $-y$ in the subsequent feedback design.

In the next section, existing classical perturbation based ES approach will be reviewed to give an idea about our proposed design and to later use in simulation comparisons.

\section{Classical Perturbation Based Extremum Seeking for Dynamic Systems} \label{pes}
In the classical ES approach shown in Fig.\ref{gradient_MISO}, a high pass filter, a multiplier, and a lowpass filter are used to find the extremum. A general single input nonlinear system is considered in the design of \cite{krsticwang}. A multi input ES approach is examined in \cite{ghaffarikrstic}. 
\begin{figure}[h!] 
\centering
\includegraphics[width=0.45 \textwidth , height=0.27 \textwidth]{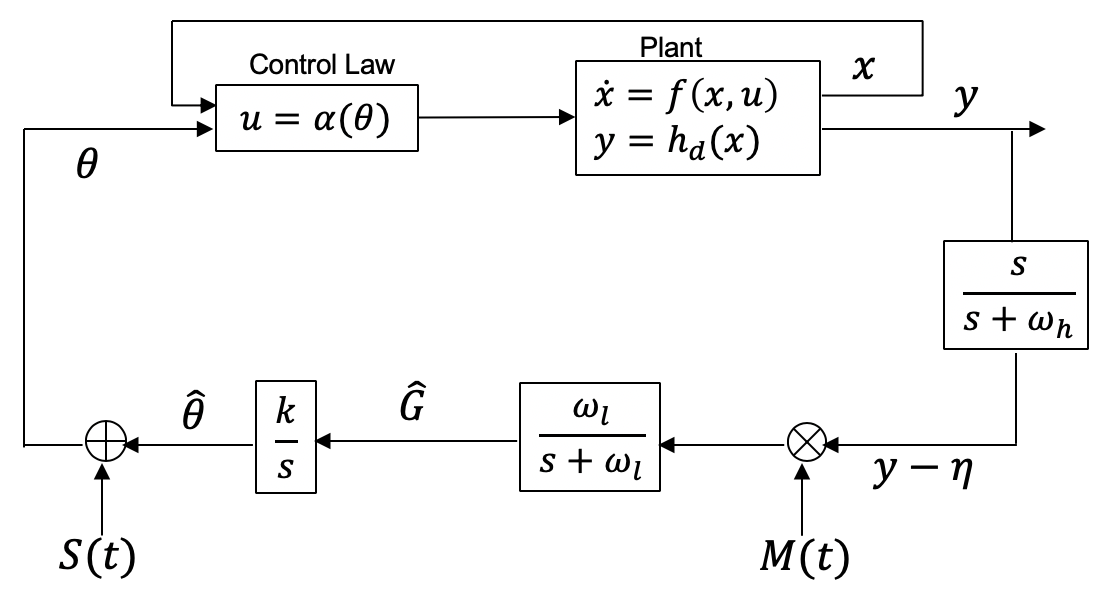}
\caption{Classic perturbation based ES scheme for multi input dynamic systems given by \cite{ghaffarikrstic}.}
\label{gradient_MISO}
\end{figure}
In the approach in \cite{krsticbook,krsticwang}, the control law (\ref{systemforlsinput}) feeding the plant (\ref{systemforls}) is tuned via the time-varying parameter $\theta=\begin{bmatrix} \theta_1, \theta_2,\cdots,\theta_N \end{bmatrix}^T$ that is produced by
\begin{equation} \label{MISO}
\begin{aligned}
 \theta (t)=\hat{\theta} (t)+S(t),
\end{aligned}
\end{equation}  
where 
\begin{equation}  \label{grad1}
\begin{aligned}
S(t)&=\begin{bmatrix} a_{1}sin(\omega_{1} t) & a_{2}sin(\omega_{2} t) & \cdots & a_{N}sin(\omega_{N} t) \end{bmatrix} ^{T},
\end{aligned}
\end{equation} 
and $\hat{\theta}(t)$ is generated by
\begin{equation}  \label{grad2}
\begin{aligned}
\dot{\hat{\theta}} (t)&=k\hat{G} (t), &\\ \dot{\hat{G}} (t)&=\omega_{l} M(t) (y (t)-\eta (t))-\omega_{l}\hat{G}(t) , & \\  \dot{\eta} (t)&=\omega_{h}\left(y(t)-\eta (t) \right). &
\end{aligned}
\end{equation}  
Perturbation signal is selected as
\\$M(t)=\begin{bmatrix} \frac{2}{a_{1}} sin(\omega_{1} t) & \frac{2}{a_{2}}sin(\omega_{2} t) & ... & \frac{2}{a_{N}}sin(\omega_{N} t) \end{bmatrix} ^{T}$. 
\\In the next two sections, we develop RLS estimation based ES scheme with forgetting factor instead of the approach of Section 3. Our proposed RLS estimation based adaptive ES scheme will be separately developed for two cases: for scalar parameter $(N=1)$ systems and for vector parameter $(N>1)$ systems, in Sections 4 and 5, respectively. 

\section{RLS based ES Design for Scalar Parameter Systems}
\subsection{Static Maps} \label{smap}
Consider the static map (\ref{staticmap_ls}) and the control law (\ref{systemforlsinput}) for scalar case, $N=1$, under Assumptions \ref{assump_1} and \ref{assump_2} about the closed-loop system. The proposed scheme is depicted in Fig. \ref{esc_ls_block_staticmap}.

RLS estimation based ES block shown in Fig. \ref{esc_ls_block_staticmap}  consists of two parts: an RLS based adaptive parameter identifier estimating the gradient $h_\theta = \frac{\partial y}{\partial \theta}$ and a control law to be fed by this estimate.


\begin{figure}[h!] 
\centering
\includegraphics[width=0.4 \textwidth , height=0.15 \textwidth]{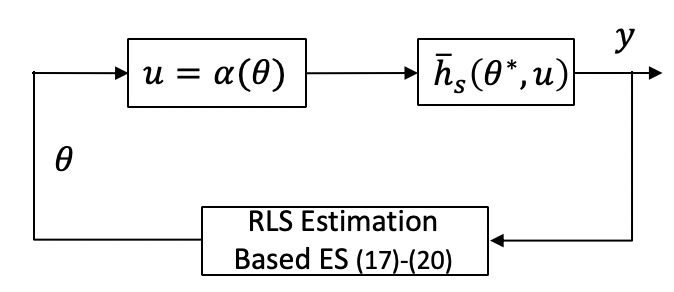}
\caption{RLS based ES scheme for scalar parameter static maps.}
\label{esc_ls_block_staticmap}
\end{figure}

Consider the static map equation (\ref{staticmap_ls}). In this equation, the time derivative of the output $y$ is given by 

\begin{equation} \label{staticoutput}
\dot{y}=h_\theta \dot{\theta}.
\end{equation}
Design of the RLS based estimator to generate $\hat{h}_\theta$ considers  the relation (\ref{staticoutput}) that is in the linear parametric model form. 
\begin{equation} \label{lpm_filter_staticmap}
z=h_\theta \phi.
\end{equation}
where  
\begin{equation}\label{eq:z_ydot}
z=\dot{y}, \quad \phi=\dot{\theta}.     
\end{equation} 
If $\dot{y}$ is not available for measurement, then the regressor signals can be generated as
\begin{equation} \label{z,phi}
\begin{aligned}
z=\frac{s}{s+\omega_l} [y], \quad \phi= \frac{1}{s+\omega_l}[\dot{\theta}],
\end{aligned}
\end{equation}
i.e.,
\begin{equation}
\begin{aligned}
\dot{z}= -\omega_l z + \dot{y}, \quad \dot{\phi}= -\phi \omega_l +\dot{\theta},
\end{aligned}
\end{equation}
where $\omega_l > 0$ is a constant design parameter.  The control law generating $\theta$ is proposed to be 
\begin{equation} \label{siso_static_rls}
\dot{\theta}=k\hat{h}_\theta, \quad k>0.
\end{equation}
Assuming that the time variation of $h_\theta$ is sufficiently slow, we design an RLS estimator for the parametric model  (\ref{lpm_filter_staticmap}) as follows:
\begin{equation} \label{siso_static_rls1}
\dot{\hat{h}}_\theta = p \epsilon \phi,
\end{equation}
\begin{equation}  \label{siso_static_rls2}
\dot{p} = \beta p - p^2 \phi^2,
\end{equation}
\begin{equation}  \label{siso_static_rls3}
\epsilon = z-\hat{h}_\theta \phi, 
\end{equation}
where $\beta > 0$ is forgetting factor and $p$ is the covariance term. The overall ES scheme producing $\theta(t)$ can be summarized by (\ref{siso_static_rls}), (\ref{siso_static_rls1}), (\ref{siso_static_rls2}), and (\ref{siso_static_rls3}).
\subsection{Dynamic Systems} \label{dsystem}
The RLS estimation based ES control scheme (\ref{siso_static_rls})-(\ref{siso_static_rls3}) applies to the dynamic system (\ref{systemforls})-(\ref{y_ls2}) for $N=1$ with the control law (\ref{systemforlsinput}) under Assumptions \ref{assump3} and \ref{assump5}. The proposed ES scheme is depicted in Fig. \ref{esc_ls_block}. 

\begin{figure}[h!] 
\centering
\includegraphics[width=0.4 \textwidth , height=0.2 \textwidth]{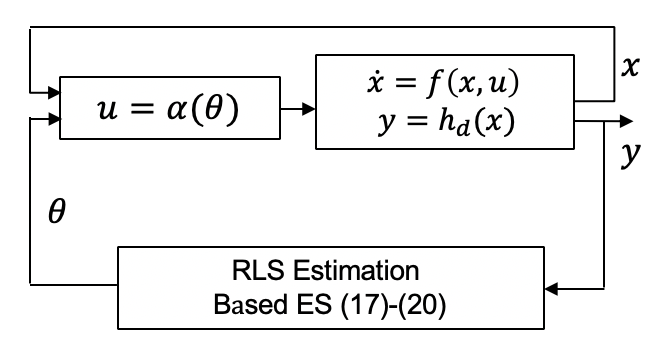}
\caption{RLS based ES scheme for scalar parameter dynamic systems.}
\label{esc_ls_block}
\end{figure}

\subsection{Stability Analysis} \label{stability_section}
In this section, stability proof of the proposed schemes in Sections \ref{smap} and \ref{dsystem} will be presented. We know that $\theta^*$ is the equilibrium point and the estimated gradient will be $h_\theta=0$ at the equilibrium point $\theta=\theta^*$. We can write our stability result as follows:
\begin{theorem}
Consider the RLS estimation based ES scheme given in Figs. \ref{esc_ls_block_staticmap}, \ref{esc_ls_block} and defined in (\ref{siso_static_rls}) - (\ref{siso_static_rls3}) with $z$ and $\phi$ as given in \eqref{eq:z_ydot} or (\ref{z,phi}), and Assumptions \ref{assump_1} - \ref{assump5}. For any initial condition $\hat{\theta}(0) \in \mathbb{R}^N$ and adaptation gain $k$, $\theta(t)$ asymptotically converges to small neighborhood of extremum parameter $\theta^*$.

%
\end{theorem}
\begin{proof}
We consider the Lyapunov function as
\begin{equation}
V(\theta(t))=\frac{1}{2}\left( \theta(t)-\theta^* \right)^2 = \frac{1}{2}\tilde{\theta}^2.
\end{equation}
We write the time derivative of $V$ along the solutions of (\ref{siso_static_rls}) as
\begin{equation} \label{lyapunov}
\dot{V}=\dot{\theta} \left( \theta(t)-\theta^* \right)=\dot{\theta} \tilde{\theta}.
\end{equation}
Substituting (\ref{siso_static_rls}) into (\ref{lyapunov}), we obtain 
\begin{equation} \label{lyapunov2}
\dot{V}=k \hat{h}_\theta \tilde{\theta}.
\end{equation}
For the maximum case, $k>0$. Negative definiteness of (\ref{lyapunov2}) depends on the initial condition $\theta_0$ that determines the signs of $\hat{h}_\theta$ and $\tilde{\theta}$. If $\theta(0) < \theta^*$, then $\hat{h}_\theta >0$ and $\tilde{\theta} < 0$. On the other hand, if $\theta(0) > \theta^*$, then $\hat{h}_\theta <0$ and $\tilde{\theta} > 0$. Hence, for both cases $\dot{V} <0$. We also need to examine the forgetting factor $\beta$ and the persistent excitation (PE) of $\phi$. If $\phi$ is PE, then (\ref{siso_static_rls}) guarantees that $p \in \mathcal{L}_{\infty}$ and $\theta(t)\to\theta^{*}$ as $t\to\infty$.  When $\beta>0$, the convergence of $\theta(t)\to\theta^{*}$ is exponential (\cite{fidan}).



\end{proof}

\section{RLS based ES Design for Vector Parameter Systems}
In this section, the proposed RLS estimation based ES scheme is extended to the systems with vector parameters $(N>1)$. Similar to the classical gradient based analysis, small sinusoidal perturbation signals with different frequencies ($\omega_{1},\cdots,\omega_{N}$) are added to the control signals to provide sufficiently rich excitation.
\subsection{Static Maps}
\begin{figure}[h] 
\centering
\includegraphics[width=0.4 \textwidth , height=0.2 \textwidth]{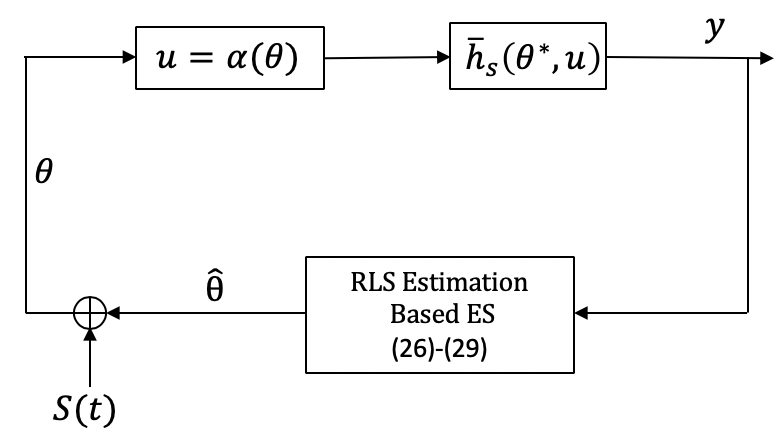}
\caption{RLS based ES scheme for vector parameter static maps.}
\label{esc_ls_block_miso_staticmap}
\end{figure}
Consider the block diagram in Fig. \ref{esc_ls_block_miso_staticmap} for the static map in (\ref{staticmap_ls}). The time derivative of (\ref{staticmap_ls}) is given by
\begin{equation} \label{staticoutput_miso}
\dot{y}=h_\theta^T \dot{\theta},
\end{equation}
which, similarly to \eqref{lpm_filter_staticmap}, can be written in the linear parametric form
\begin{equation}\label{eq:SPMvector}
z=h_\theta^T \phi,     
\end{equation}  
where $z$ and $\phi$ are again defined by either \eqref{eq:z_ydot} or (\ref{z,phi}). The control law \eqref{siso_static_rls} is used for updating $\theta$ in the vector case as well. 
The design of the RLS estimator to produce $\hat{h}_\theta$ is based on the parametric model  \eqref{eq:SPMvector} and is given as follows (\cite{fidan}):
\begin{equation} \label{siso_static_rls11}
\dot{\hat{h}}_\theta=P\epsilon\phi,
\end{equation}
\begin{equation}  \label{siso_static_rls21}
\dot{P} = \beta P - P \phi \phi^T P,
\end{equation}
\begin{equation}  \label{siso_static_rls31}
\epsilon = z-\hat{h}_\theta^T \phi, 
\end{equation}
where $\beta$ is the forgetting factor and $P$ is the covariance matrix of the RLS algorithm. The control law generating $\theta$ is proposed to be
\begin{equation} \label{miso_controllaw}
\dot{\hat{\theta}}=k\hat{h}_\theta, \quad k>0. 
\end{equation}
\begin{equation} \label{miso}
 \theta(t)=\hat{\theta}(t)+S(t),
\end{equation}
where $S(t)$ is defined as in \eqref{grad1}. Different from scalar parameter systems, we use perturbation signals, $S(t)$. The need to use of dither signals in vector parameter systems is that dither signals with different frequencies can be implemented on each input signal to achieve overall PE. 
\subsection{Dynamic Systems}
\begin{figure}[h!] 
\centering
\includegraphics[width=0.4 \textwidth , height=0.23 \textwidth]{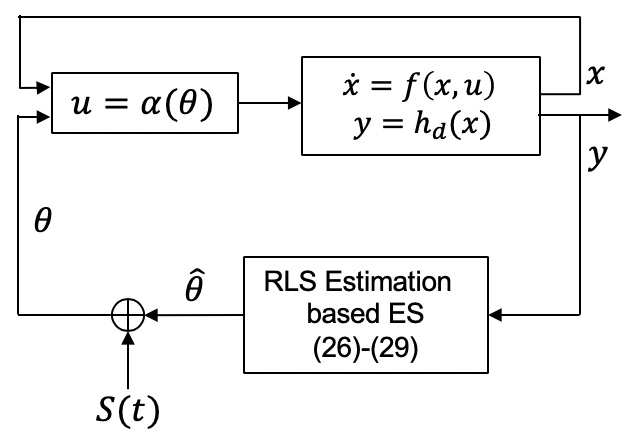}
\caption{RLS based ES scheme for vector parameter dynamic systems.}
\label{esc_ls_block_miso}
\end{figure}
The RLS estimation based ES scheme  (\ref{siso_static_rls11}) - (\ref{miso_controllaw}) applies to the dynamic system (\ref{systemforls})-(\ref{y_ls2}) with control law (\ref{systemforlsinput}) under Assumptions \ref{assump3} and \ref{assump5} for vector parameter systems. Block diagram of the proposed ES scheme is given in Fig.\ref{esc_ls_block_miso}. 
\subsection{Stability Analysis} 
The intuition in \eqref{miso} is to satisfy persistence of excitation for $N$-dimensional $\phi$ by introducing at least one distinct dither frequency for each input, following the standard perturbation based ES control approaches mentioned in Section \ref{pes}. Similar to the analysis in Section \ref{stability_section}, consider the Lyapunov function as
\begin{equation}
V(\tilde{\theta}(t))=\frac{1}{2}\tilde{\theta}^T\tilde{\theta}.
\end{equation}
We write the time derivative of $V$ along the solutions of (\ref{miso_controllaw}) as
\begin{equation} \label{lyapunovmiso}
\dot{V}=\tilde{\theta}^T  \dot{\tilde{\theta}} = \tilde{\theta}^T  \dot{\theta}.
\end{equation}
Substituting (\ref{miso}) into (\ref{lyapunovmiso}), we obtain 
\begin{equation} \label{lyapunov2miso}
\dot{V}=\tilde{\theta}^T (k \hat{h}_\theta +\dot{S}) .
\end{equation}
The relationship between $\tilde{\theta}$ and $ \hat{h}_\theta$ in Section 4.3 applies to vector parameter case. The stability again depends on $k$, initial condition $\theta(0)$, forgetting factor $\beta$, and PE of $\phi$, that is guaranteed by addition of dither signals in (\ref{miso}). Hence, $P \in \mathcal{L}_{\infty}$ and $\theta(t)\to\theta^{*}$ as $t\to\infty$. 
\section{Simulations}
In this section, we present simulation results to show the validity of the proposed schemes. We will present two examples for scalar parameter and vector parameter cases with their comparison results with classical ES method in Section \ref{pes}.
\subsection{Scalar Parameter Simulation Example}
Consider the following model
\begin{equation}
\begin{aligned}
y&=10m(u), &\\ m(u)&=k_1 \left( 1-e^{-k_2 u} \right)-k_3 u  &\\ u&=\theta,& 
\end{aligned}
\end{equation}
where $\theta^*=0.3$.  $\theta_0=0.01$ is chosen as initial value for both schemes. $k_1=1.05, k_2=23, k_3=0.52$ are given. For RLS estimation based ES scheme, the following parameters are used: $k_{ls}=0.01$, $p_0=10^3$, and $\beta=0.98$ are given.  For classical ES scheme, the following parameters are given: $k=0.08$, $\omega_h=0.6$, $\omega_l=0.8$, $S(t)=0.01\sin3t$, and $M(t)=sin3t$. We apply the Gaussian measurement noise as ($\sigma=0.05$) for both gradient and RLS algorithms.
\noindent
We apply RLS estimation based ES scheme in Fig.\ref{esc_ls_block}. The results for this example is given in Fig.\ref{y_result}. It is obvious that  proposed scheme can reach a neighborhood of the extremum point $\theta^*=0.3$ at $y^*=8.85$ less than 2 second while classical ES finds the extremum point very late and cannot maintain that extremum point under measurement noise. 
\begin{figure}[h!] 
\centering
\includegraphics[width=0.44 \textwidth , height=0.25 \textwidth]{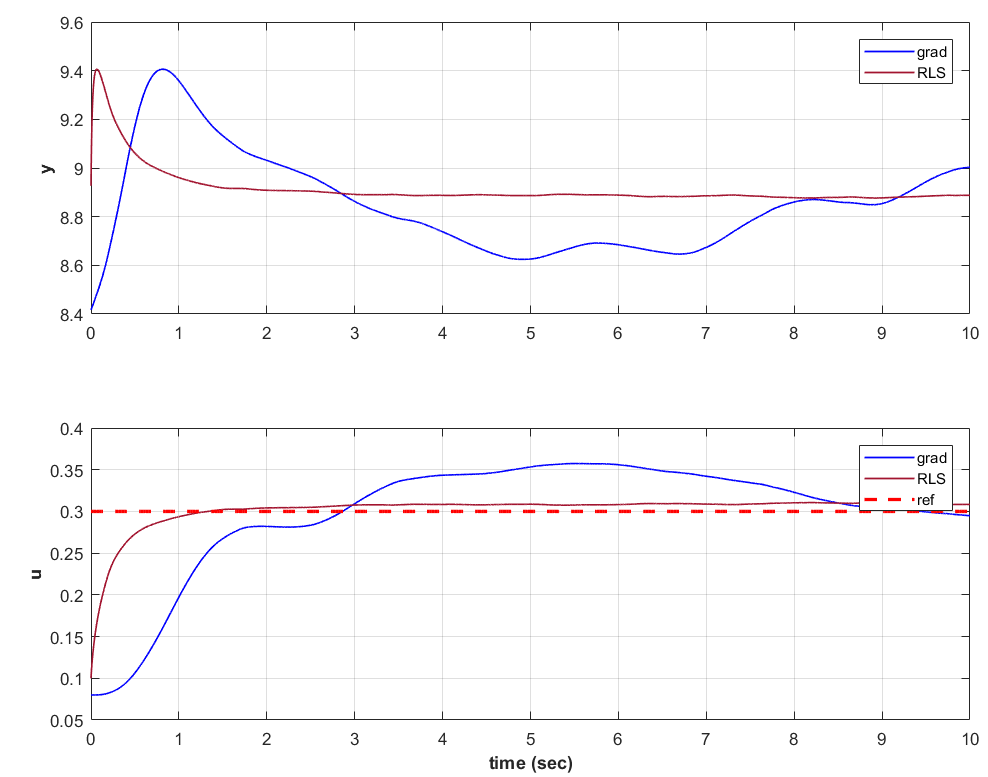}
\caption{Single parameter RLS estimation based ES results.}
\label{y_result}
\end{figure}
\subsection{Vector Parameter Simulation Example}
Consider the following model 
\begin{equation}
\begin{aligned}
y&=y_1+y_2, &\\ y_1&=am(u_1), ~ m(u_1)= (2m^*_1 u^*_1 u_1)/(u^{*2}_1+ u^2_1), &\\ y_2&=am(u_2), ~ m(u_2)= (2m^*_2 u^*_2 u_2)/(u^{*2}_2+ u^2_2),&\\ u&=[u_1,~u_2]=[\theta_1, ~ \theta_2].& 
\end{aligned}
\end{equation}

\noindent
where $[\theta^*_1, ~ \theta^*_2]=[0.2, ~0.3]$. For both schemes, initial values are given as $u_0=[0.1 ,~ 0.1].$ We aim to reach $y^*_1 (\theta^*_1)=5$ and $y^*_2 (\theta^*_2)=9$. For RLS estimation based ES scheme, the following parameters are used: $k=[ 0.01, ~ 0.01 ]$, $P_0=10^4$, $\beta=0.98$, and $S(t)=[ 0.01 \sin 7t, ~ 0.01 \sin 10t ]$ are given.  For classical ES scheme, the following parameters are given: $k=[0.02, ~0.01]$, $\omega_h=[ 0.6,~ 0.6 ] $, $\omega_l=[ 0.8, ~0.8 ]$, $S(t)=[ 0.01 \sin t,~ 0.01 \sin 2t]$, and $M(t)=[ 4.5 \sin 5t, ~11 \sin 5t ]$. We apply the Gaussian measurement noise as ($\sigma=0.05$) for both gradient and RLS algorithms.
\noindent
Simulation results are given in Fig.\ref{y_miso} for both RLS estimation based and classical ES schemes. It is clear that the results taken with RLS  can converge the extremum point and find the maximized output $y^*$ while classical ES scheme has difficulty to reach the extremum point. One reason for this difficulty is that in classical ES scheme has many tuning parameters that must be tuned accordingly.  
\noindent
For vector case, we also emphasize the need to apply perturbation terms to the scheme in order to observe multiple input channels separately. When there is no perturbation signal applied, the inputs cannot be distinguished and converge to an average value that caused to reach a value near the maximum. Similar to scalar case, RLS estimation based ES scheme outweighs classical ES scheme in terms of reaching extremum under measurement noises.

\subsection{ABS Simulation Example}
In this section, we also tested our ES scheme in ABS using MATLAB/Simulink. Then, we compared its performance with gradient based ES scheme developed by \cite{krsticbook}. The wheel characteristics are given by the following set of equations
\begin{figure}[h!] 
\centering
\includegraphics[width=0.44 \textwidth , height=0.35 \textwidth]{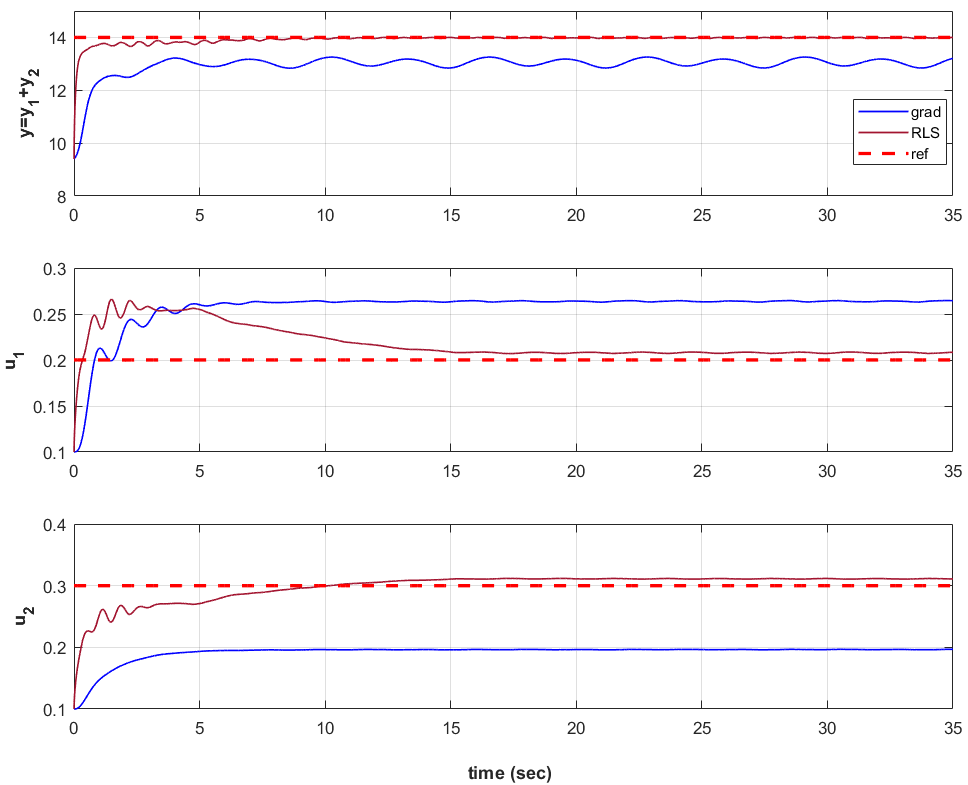}
\caption{Vector parameter RLS estimation based ES results.}
\label{y_miso}
\end{figure}
\begin{equation} \label{model}
\begin{aligned}
m\dot{\nu}&=N\mu(\lambda), &\\ I\dot{\omega}&=-B\omega - N R \mu(\lambda)+\tau, &
\end{aligned}
\end{equation}
where $v, \omega, m, N, R, I$ are linear velocity, angular velocity, the mass, the weight, radius, and the moment of inertia of the wheel, respectively. $B\omega$ is the bearing friction torque, $\tau$ is braking torque, $\mu(\lambda)$ is the friction force coefficient. 
\noindent
$\lambda$ is the wheel slip which is defined as 
\begin{equation} \label{wheel}
\lambda(v,\omega)=\frac{R\omega-\nu}{\nu}.
\end{equation}
Controller design procedure are identical to the design in \cite{krsticbook}. 
The parameters that are identical in both schemes are given as follows: $m=400 kg$, $R=0.3 m$, $I=1.7 kgm^2$, $B=0.01 kg/s$. Perturbation signal amplitude and frequency is selected as $a=0.01$,  $\omega=3$, high pass, low pass and regulation gain are selected as $\omega_{h}=0.6$, $\omega_{l}=0.8$, $k=6$ in gradient based scheme equations (\ref{MISO}), (\ref{grad1}), and (\ref{grad2}).
\noindent
$k=-0.01$ is used in ABS case and $\beta=0.95$ is selected for RLS based scheme. The simulation for both gradient and RLS schemes is performed under the Gaussian noise ($\sigma=0.1$) in longitudinal acceleration measurement, $\dot{v}$. Initial conditions are selected the same in both schemes for a fair comparison. We use the approximation model (\ref{krsticmu}) in simulations to see the effect of the proposed schemes.
\begin{equation} \label{krsticmu}
\mu (\lambda)=2\mu_{max}\frac{\lambda^{*} \lambda}{{{\lambda}^{*2}}+\lambda^{2}},
\end{equation}
\noindent
where (\ref{krsticmu}) has a maximum at $\lambda=\lambda^{*}$ with $\mu(\lambda^*)=\mu_{m}$.
\noindent
For simulation, we choose wet road since it is one of the safety critical conditions. Simulation results of ABS for gradient/RLS based scheme comparison are given in Fig.\ref{ABS}. Results show that vehicle stopping time of RLS parameter estimation based ES in an emergency situation is less than that of gradient one. Slip ratio estimation is almost 2 sec quicker with RLS parameter estimation, can be seen in Fig. \ref{ABS_dry_mu_lambda}. RLS based ES scheme gives better results under measurement noise and can reach the maximum deceleration in less time.
\begin{figure}[h!] 
\centering
\begin{subfigure}{0.47\textwidth}
\includegraphics[width=1 \textwidth , height=0.78\textwidth]{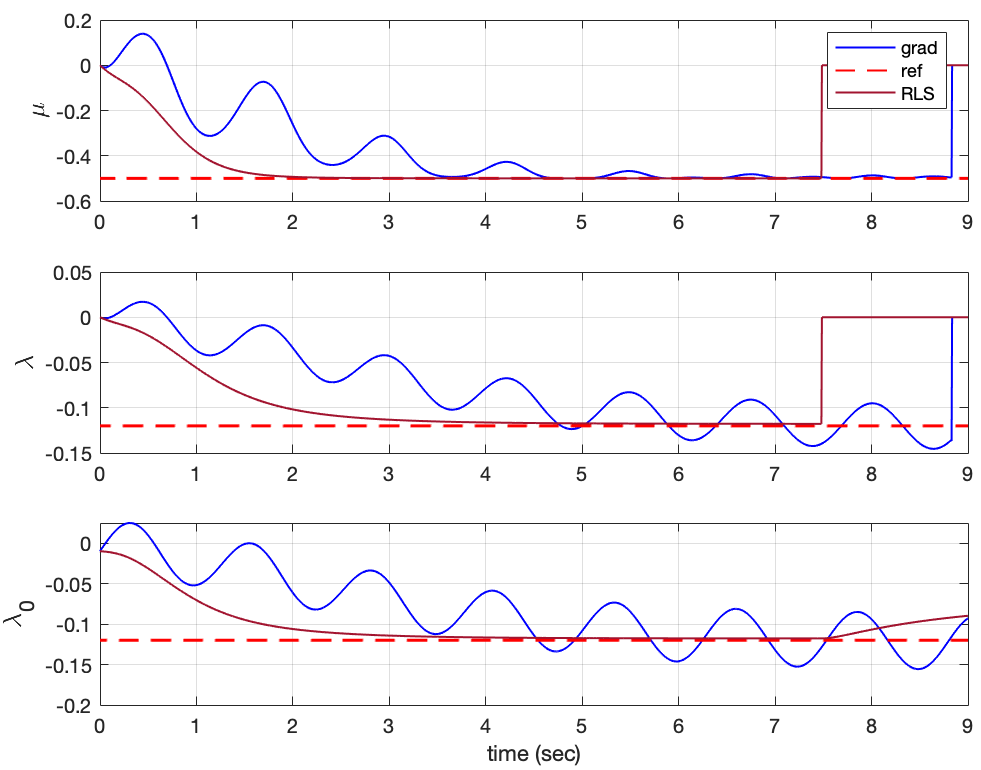}
\caption{Friction force coefficient and estimated slip results for ABS.}
\label{ABS_dry_mu_lambda}
\end{subfigure}
~
\begin{subfigure}{0.47\textwidth}
\includegraphics[width=1 \textwidth , height=0.78 \textwidth]{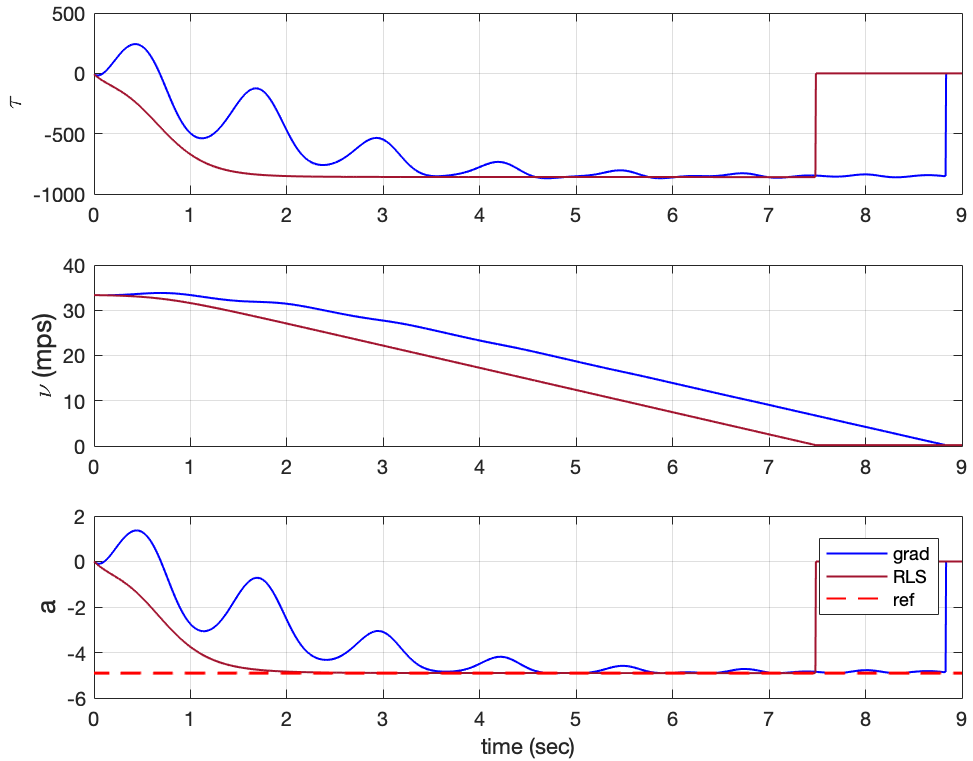}
\caption{Braking torque, velocity and deceleration results for ABS.}
\label{ABS_dry_tau}
\end{subfigure}
    \caption{Wet road comparison results for ABS.}
        \label{ABS}
\end{figure}
\section{Conclusion}
This paper focuses on designing an RLS parameter estimation based ES scheme for scalar parameter and vector parameter static map and dynamic systems. Their stability conditions are stated for each case. The proposed ES scheme does not need perturbation signals for scalar parameter systems; however, the proposed ES scheme needs perturbation signals with different frequencies for vector parameter systems. Proposed scheme is applied to different simulation scenarios and compared to classical gradient estimation based ES under measurement noise. The results show the validity and effectiveness of RLS parameter estimation based ES scheme over gradient one.


   
\bibliographystyle{IEEEtran}
\bibliography{Zengin_cdc2020_adaptive_ES}   
                                                   







\end{document}